\documentclass{amsart}

\usepackage{amssymb, amscd,stmaryrd,setspace,hyperref,color}

\input xy
\xyoption{all} \xyoption{poly} \xyoption{knot}\xyoption{curve}

\newcommand{\comment}[1]{}

\comment{The following is Eli Lebow's trick for making the table of contents include all definitions and labels.

\setcounter{tocdepth}{5}

}

\newcommand{\beqn}[1]{\begin{eqnarray}\label{#1}}
\newcommand{\eeqn}{\end{eqnarray}}

\def\tn{\textnormal}

\def\mc{\mathcal}

\def\PP{{\mathbb P}}
\def\NN{{\mathbb N}}

\def\Hom{\tn{Hom}}
\def\Ob{\tn{Ob}}

\def\to{\rightarrow}

\def\cross{\times}
\def\taking{\colon}

\def\tto{\rightrightarrows}

\def\ss{\subset}

\def\iso{\cong}
\def\|{{\;|\;}}
\def\m1{{-1}}
\def\op{^\tn{op}}

\def\ul{\underline}

\def\ullimit{\ar@{}[rd]|(.3)*+{\lrcorner}}
\def\urlimit{\ar@{}[ld]|(.3)*+{\llcorner}}
\def\lllimit{\ar@{}[ru]|(.3)*+{\urcorner}}
\def\lrlimit{\ar@{}[lu]|(.3)*+{\ulcorner}}
\def\ulhlimit{\ar@{}[rd]|(.3)*+{\diamond}}
\def\urhlimit{\ar@{}[ld]|(.3)*+{\diamond}}
\def\llhlimit{\ar@{}[ru]|(.3)*+{\diamond}}
\def\lrhlimit{\ar@{}[lu]|(.3)*+{\diamond}}
\newcommand{\clabel}[1]{\ar@{}[rd]|(.5)*+{#1}}

\newcommand{\arr}[1]{\ar@<.5ex>[#1]\ar@<-.5ex>[#1]}
\newcommand{\arrr}[1]{\ar@<.7ex>[#1]\ar@<0ex>[#1]\ar@<-.7ex>[#1]}
\newcommand{\arrrr}[1]{\ar@<.9ex>[#1]\ar@<.3ex>[#1]\ar@<-.3ex>[#1]\ar@<-.9ex>[#1]}
\newcommand{\arrrrr}[1]{\ar@<1ex>[#1]\ar@<.5ex>[#1]\ar[#1]\ar@<-.5ex>[#1]\ar@<-1ex>[#1]}

\newcommand{\To}[1]{\xrightarrow{#1}}

\newcommand{\push}[4]{\xymatrix{#1\ar[r]\ar[d] \ar@{}[rd]|(.7)*+{\lrcorner} & #2 \ar[d] \\ #3 \ar[r] & #4}}
\newcommand{\Push}[8]{\xymatrix{#1\ar[r]^-{#5}\ar[d]_-{#6} \ar@{}[rd]|(.7)*+{\lrcorner} & #2 \ar[d]^-{#7} \\ #3 \ar[r]_-{#8} & #4}}
\newcommand{\pull}[4]{\xymatrix{#1\ar[r]\ar[d] \ar@{}[rd]|(.3)*+{\ulcorner} & #2 \ar[d] \\ #3 \ar[r] & #4}}
\newcommand{\Pull}[8]{\xymatrix{#1\ar[r]^-{#5}\ar[d]_-{#6} \ar@{}[rd]|(.3)*+{\ulcorner} & #2 \ar[d]^-{#7} \\ #3 \ar[r]_-{#8} & #4}}
\newcommand{\hpush}[4]{\xymatrix{#1\ar[r]\ar[d] \ar@{}[rd]|(.7)*+{\diamond} & #2 \ar[d] \\ #3 \ar[r] & #4}}
\newcommand{\hPush}[8]{\xymatrix{#1\ar[r]^{#5}\ar[d]_{#6} \ar@{}[rd]|(.7)*+{\diamond} & #2 \ar[d]^{#7} \\ #3 \ar[r]_{#8} & #4}}
\newcommand{\hpull}[4]{\xymatrix{#1\ar[r]\ar[d] \ar@{}[rd]|(.3)*+{\diamond} & #2 \ar[d] \\ #3 \ar[r] & #4}}
\newcommand{\hPull}[8]{\xymatrix{#1\ar[r]^-{#5}\ar[d]_-{#6} \ar@{}[rd]|(.3)*+{\diamond} & #2 \ar[d]^-{#7} \\ #3 \ar[r]_-{#8} & #4}}

\newcommand{\adjoint}[2]{\xymatrix@1{#1\ar@<.5ex>[r] & #2 \ar@<.5ex>[l]}}
\newcommand{\Adjoint}[4]{\xymatrix@1{#2 \ar@<.5ex>[r]^-{#1} & #3 \ar@<.5ex>[l]^-{#4}}}

\def\Sets{{\bf Sets}}
\def\sSets{{\bf sSets}}

\def\mcB{\mc{B}}
\def\mcC{\mc{C}}
\def\mcD{\mc{D}}

\def\mcG{\mc{G}}
\def\mcH{\mc{H}}
\def\mcI{\mc{I}}
\def\mcJ{\mc{J}}

\def\mcS{\mc{S}}

\newtheorem{lemma}[subsection]{Lemma}
\newtheorem{proposition}[subsection]{Proposition}

\theoremstyle{remark}
\newtheorem{remark}[subsection]{Remark}
\newtheorem{example}[subsection]{Example}

\newtheorem{question}[subsection]{Question}
\newtheorem{guess}[subsection]{Guess}

\theoremstyle{definition}
\newtheorem{definition}[subsection]{Definition}

\def\ss{\subseteq}

\def\bD{{\bf \Delta}}
\def\OO{{\mathbb O}}

\begin{document}
\begin{abstract}

Networks are often studied as graphs, where the vertices stand for entities in the world and the edges stand for connections between them.  While relatively easy to study, graphs are often inadequate for modeling real-world situations, especially those that include contexts of more than two entities.  For these situations, one typically uses hypergraphs or simplicial complexes.

In this paper, we provide a precise framework in which graphs, hypergraphs, simplicial complexes, and many other categories, all of which model higher graphs, can be studied side-by-side.  We show how to transform a hypergraph into its nearest simplicial analogue, for example.  Our framework includes many new categories as well, such as one that models broadcasting networks.  We give several examples and applications of these ideas.

\end{abstract}

\author{David I. Spivak}

\title{Higher-dimensional models of networks}

\thanks{This project was supported in part by a grant from the Office of Naval Research: N000140910466.}

\maketitle

\tableofcontents

\section{Introduction}\label{sec:intro}

Researchers studying networks have a variety of mathematical objects which they may use as models.  The most basic among these are directed or undirected graphs (\cite{Bro},\cite{Bol},\cite{Jac}), which consist of a set of nodes and edges.  The nodes represent entities, and the edges represent lines of contact or communication between them.  This conception works well in many circumstances, but when its purpose is to model how communication works in the real world, it suffers from a glaring shortcoming: communication can take place in groups of more than two entities.  

Let us imagine a table of four people having a conversation together at a restaurant; this event is understood as a social network with four nodes.  If everyone can hear everyone else, then graph-theoretically this network is represented by a complete graph on four vertices.  Now imagine a situation of four people playing the game ``telephone" so that each person may only whisper in the ear of another person.  If each person can whisper to anyone else, this scenario is again modeled graph-theoretically by a complete graph on four vertices.  But the situations are extremely different!  One-dimensional graph theory does not capture the distinction between a single 4-person conversation and six 2-person conversations. 

To study such situations, one must turn from (1-dimensional) graphs to a higher-dimensional model.  The two most popular such models are hypergraphs \cite{DW} and simplicial sets \cite{Atk}.  A hypergraph is like a graph, except that edges can connect more than two vertices.   One draws a hypergraph by drawing a set of dots and enclosing certain subsets of them within circles.   Simplicial sets, on the other hand, are visualized as spaces, more specifically as multi-dimensional polygonal shapes made up of vertices, edges, triangles, and higher-dimensional triangles, like tetrahedra.  An $n$-dimensional triangle (called an {\em $n$-simplex}) can be thought of as the ``polyhedral hull" of $n+1$ vertices.  Thus, a 1-dimensional triangle, or 1-simplex, is the polyhedral hull of 2 vertices: it is simply an edge.  Likewise, a 2-simplex is the hull of 3 vertices and is hence a triangle; a 3-simplex is a tetrahedron; and a 0-simplex is just a vertex.  

In general, a simplicial set is the union of many such vertices, edges, triangles, etc.  To make this a little more clear, let us emphasize that any graph is a 1-dimensional simplicial set.  Now suppose we want to make a 2-dimensional simplicial set.  Begin with a graph $G$, and choose a set of three edges $(a,b),(b,c)$, and $(a,c)$ which form a triangle inside $G$:

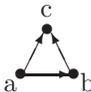
\begin{center} \begin{picture}(18,27)\put(-4,-4){a}\put(25,-4){b}\put(10,24){c}\put(0,0){$\bullet$}\put(3,2){\vector(1,0){18}}\put(3,2){\vector(1,2){8}}\put(20,0){$\bullet$}\put(10,17){$\bullet$}\put(22,2){\vector(-1,2){8}}\end{picture}\end{center}

\noindent Given such a triangle, one may attach in a 2-simplex to $G$, filling in the triangle $abc$.  Thus a 2-dimensional simplicial set looks like a graph except that some triangles are filled.  To create a 3-dimensional simplicial set, begin with a 2-dimensional simplicial set, choose some groupings of four triangles that form an empty tetrahedron, and fill them.  At this point, the reader can imagine constructing arbitrary simplicial set.  An $n$-dimensional simplicial set is one  in which the largest simplex is an $n$-simplex.

The appropriate simplicial set model for our ``table of four" is the tetrahedron.  This solid shape represents the idea that the conversation is taking place between four entities in a {\em shared space.}  For the situation in which the four people can only whisper to each other, we instead use the simplicial set consisting only of the six edges of the tetrahedron.  One checks that this graph is the complete graph on four vertices, as above.  

Simplicial complexes (a cousin of simplicial sets) were first offered as models for social networks by Atkins \cite{Atk}, and his theory of Q-analysis has been used for some time (\cite{Fre},\cite{JY},\cite{ZTB}).  Unfortunately, simplicial complexes have a major drawback, as compared with simplicial sets, for studying networks.  Namely, no two contexts can have the same set of participants (see Remark \ref{rem:drawbacks of complexes}).  Moreover, definitions of simplicial complexes and hypergraphs, especially their morphisms are not always consistent or precise in the literature.  Neither are methods for transforming a hypergraph into its ``nearest" simplicial complex and vice versa.  Finally, it seems that the simplices in a simplicial complexes are often imagined as unordered, whereas under the strict definition, they are ordered.  We believe these issues can meaningfully detract from the quality of discourse among researchers.

In this paper, we provide precise mathematical definitions of all of the above-mentioned models for higher graphs, as well as the morphisms between them.   We define a framework in which graphs, hypergraphs, and simplicial sets exist along side one-another, and we explain how one can transform objects from one into another.  This framework contains many other useful notions of higher-dimensional graphs, one of which is unordered simplicial sets.  All of this is presented in Section \ref{sec:math defs}.  We hope that the ability to change between various models, as well as to see the broad range of forms these models can take, will assist the researcher in choosing the best one for the situation at hand.

In Section \ref{sec:choosing}, we give several examples of real-world networks (including those of computer processes, military missions, and conference calls) and explain how one chooses the mathematical model which best describes each one.  Finally, in Section \ref{sec:apps} we discuss further applications of these ideas as well as directions for future research.  

\subsection{Notation}

In this paper, compositions of functions are written in the classical mathematical order, i.e. the composite of $A\To{f}B\To{g}C$ is written $g\circ f\taking A\to C$.  Disjoint unions of sets are denoted with the \;$\amalg$\; symbol (e.g. $A\amalg B$ or $\coprod_{i\in I}A_i$).  If $\mcC$ is a category, its set of objects is denoted $\Ob(\mcC)$ and if $c,d\in\Ob(C)$ are objects, we write $\Hom_\mcC(c,d)$ or just $\Hom(c,d)$ for the set of arrows or morphisms between them.  Computer scientists should consult \cite{Pie} for more on category theory.

\subsection{Acknowledgments}

The author would like to thank Paea LePendu and Masoud Valafar for many useful conversations.

\section{Mathematical definitions}\label{sec:math defs}

In this section we give rigorous mathematical definitions for graphs, hypergraphs, simplicial sets, etc.  We begin with background material (Section \ref{subsec:background}), such as the definition of graphs and of hypergraphs, and their respective morphisms.  We then proceed to give a set-theoretic definition of semi-simplicial sets, which are a slightly less sophisticated version of simplicial sets, in Section \ref{subsec:semi}. 

In Section \ref{subsec:cat defs} we switch gears and begin to view everything from a category-theoretic standpoint.  We redefine graphs, hypergraphs, and semi-simplicial sets as categories of functors, and show that these definitions match with the set-theoretic ones given above.  We explain how to compare these various models in Section \ref{subsec:comparisons}, and then finally generalize everything in this section by defining {\em categories of higher graphs} and giving some new examples in Section \ref{subsec:cats of higher graphs}.

Some researchers find precise mathematical definitions to be enlightening, while others find them tedious.  Those in the latter category should take a look at the introductory remarks in Section \ref{subsec:cat defs} before skipping to Section \ref{sec:choosing}.

\subsection{Background}\label{subsec:background}

A graph consists of a set of edges, a set of vertices, and a way to know which two vertices belong to a given edge.  We allow multiple edges to connect the same two vertices, and we allow loops.  Precisely, we have the following definition.

\begin{definition}

A {\em directed graph} is a triple $X=(X_E,X_V,f)$, where $X_E$ and $X_V$ are sets (called the set of {\em edges} and the set of {\em vertices} respectively), and where $f\taking X_E\to X_V\cross X_V$ is a function, which we call the {\em constituent function}. 

\end{definition}

A hypergraph is like a graph in that it has vertices and (hyper-) edges, but the set of vertices contained in a hyperedge can have cardinality other than 2.  For a set $S$, let $\PP_+(S)$ denote the set $\{A\ss S|A\neq\emptyset\}$ of non-empty subsets of $S$.  Recall that elements of $\PP_+(S)$ are unordered sets; we define an ordered analogue as follows: $$\OO_+(S):=\coprod_{n\geq 1} S^n=S\amalg (S\cross S)\amalg (S\cross S\cross S)\amalg\cdots,$$ the set of all non-empty tuples on $S$.  Each tuple $x\in\OO_+(S)$ has a {\em cardinality} $k=|x|$, where $x\in S^k$; we may refer to $x$ as a {\em $k$-edge}. There is an obvious map $\OO_+(S)\to\PP_+(S)$ given by sending each $n$-tuple to its set of elements; note that this map may decrease cardinality.   

\begin{definition}

A {\em directed hypergraph} is a triple $X=(X_E,X_V,f)$, where $X_E$ and $X_V$ are sets (called the set of {\em hyperedges} and {\em vertices} of $X$, respectively), and where $f\taking X_E\to\OO_+(X_V)$ is a function, which we call the {\em constituent function}. 

An {\em undirected hypergraph} is a triple $X=(X_E,X_V,f)$ with $X_E$ and $X_V$ as above, and where $f\taking E\to\PP_+(V)$ is a function, again called the {\em (unordered) constituent function}.

\end{definition}

We can see easily that a directed graph is just a directed hypergraph with no $k$-edges for $k\neq 2$.  This generalizes as follows. 

\begin{definition}

A {\em $k$-uniform hypergraph} (or simply {\em $k$-graph}) is a hypergraph $(X_E,X_V,f)$, such that $|e|=k$ for all hyperedges $e\in E$.

\end{definition}

The above definition is meant to include both directed and undirected hypergraphs, using the appropriate definitions of cardinality given above.

There is a well-defined notion of {\em morphisms} (or functions) between graphs (resp. between hypergraphs), which we give below.  These morphisms tell us what structure needs to be preserved when comparing two graphs (resp. hypergraphs).  The idea is that one should send vertices to vertices and edges to edges, in a consistent way.

In the following definition and several others, we will assume the reader knows the first few definitions of category theory, which can be found in \cite{MacLane} or \cite{Pie}.  One thinks of a category as roughly like a directed graph, the main difference being a certain transitivity law: every pair of edges $a\to b$ and $b\to c$ gives rise to a distinguished {\em composite} edge $a\to c$.  The nodes in a category are called {\em objects} and the edges are called {\em morphisms}.  Note that there is a layer of abstraction here: if a category is like a graph, then the category of hypergraphs is like a graph that has hypergraphs as nodes and functions between hypergraphs as edges. 

\begin{definition}\label{def:cat of graph}

Let $X=(X_E\To{f} X_V\cross X_V)$ and $Y=(Y_E\To{g}Y_V\cross Y_V)$ denote graphs.  A {\em morphism} $X\to Y$ consists of a function on edges $a\taking X_E\to Y_E$ and a function on vertices $b\taking X_V\to Y_V$, such that the diagram $$\xymatrix{X_E\ar[r]^-f\ar[d]_a& X_V\cross X_V\ar[d]^{b\cross b}\\Y_E\ar[r]_-g&Y_V\cross Y_V}$$ commutes.

With the above objects and morphisms, and with the obvious composition law, the {\em category of directed graphs} has now been defined; we denote it $\mcG$.

\end{definition}

\begin{definition}\label{def:cat of hyper}

Let $X=(X_E\To{f}\OO_+(X_V))$ and $Y=(Y_E\To{g}\OO_+(Y_V))$ be directed hypergraphs.  A {\em morphism} $X\to Y$ consists of functions $a\taking X_E\to Y_E$ and $b\taking X_V\to Y_V$, such that the diagram $$\xymatrix{X_E\ar[r]^-{f}\ar[d]_a&\OO_+(X_V)\ar[d]^{b_*}\\Y_E\ar[r]_-{g}&\OO_+(Y_V)}$$ commutes.  Here $b_*(v_1,\ldots v_n)=(b(v_1),\ldots,b(v_n))$.

If $X$ and $Y$ are undirected hypergraphs, we define morphisms $X\to Y$ in the same way, except that in the diagram we replace $\OO_+$ with $\PP_+$ and say that $b_*(\{v_1,\ldots,v_n\})=\{b(v_1),\ldots,b(v_n)\}$.

With the above objects and morphisms, and with the obvious composition law, the {\em category of directed ({\tn resp.} undirected) hypergraphs} has now been defined; we denote it $\mcH$.

\end{definition}

\subsection{Semi-simplicial sets}\label{subsec:semi}

We discussed simplicial sets in the introduction, but everything we said there could be said about semi-simplicial sets as well.  After giving the definition (\ref{def:set def of sset}) of semi-simplicial sets, we discuss how one should visualize them (which can be made formal) in Remark \ref{rem:geom intuition}.  A square is explicitly constructed as a semi-simplicial set in Example \ref{ex:square}.  We will briefly explicate the differences between simplicial and semi-simplicial sets in Remark \ref{rem:difference semi vs. simp}.  

\begin{definition}[Set-theoretic definition of semi-simplicial set]\label{def:set def of sset}

A {\em semi-simplicial set} $X$ consists of a sequence of sets $(X_0,X_1,X_2,\ldots)$, and for every $n\geq 1$ a collection of functions $d^n_0,d^n_1,\ldots,d^n_n$, each of which has domain $X_n$ and codomain $X_{n-1}$.  These functions must satisfy the following identity: \begin{align}\label{dia:semi identity} d^n_i\circ d^{n+1}_j=d^n_{j-1}\circ d^{n+1}_i, \hspace{.2in} \tn{for all } 0\leq i<j\leq n.\end{align}  

For each $n\geq 0$ the set $X_n$ is called the set of {\em $n$-simplices} of $X$.  The function $d^n_i\taking X_n\to X_{n-1}$ is often abbreviated $d_i$ and is called the {\em $i$th face operator}.  The {\em dimension} of $X$ is the largest number $n$ such that $X_n\neq\emptyset$, if a largest such $n$ exists, and $\infty$ otherwise.

We may write $X$ as a sequence of sets and face maps $$\xymatrix{\cdots \arrrr{r}&X_2\arrr{r}&X_1\arr{r}&X_0.}$$

\end{definition}

\begin{remark}\label{rem:geom intuition}

Definition \ref{def:set def of sset} may be difficult to comprehend on first glance.  One should think of $X$ as a polygonal shape made of triangles, like two triangles put together to form a square (see Example \ref{ex:square}).  The set $X_0$ is the set of vertices of that shape, the set $X_1$ is the set of edges, the set $X_2$ is the set of triangles (called 2-simplices in this parlance because they are 2-dimensional), etc.  

Each 2-simplex in $X$ has three edges: across from each vertex in the triangle is an edge not touching it.  The three face operators $d^2_0,d^2_1,d^2_2$ are tasked with telling us which three edges form the boundary of a given 2-simplex.  Going further and applying $d^1_0$ and $d^1_1$ to each of these would then give us vertices for each 2-simplex.  

The above description can be made formal, in the sense that there is a {\em geometric realization} functor from the category of semi-simplicial sets to the category of topological spaces.  See \cite{Hov}.

\end{remark}

\begin{remark}

Implementing finite semi-simplicial sets on a computer is not really harder than implementing (labeled) finite directed graphs.  One inputs a set of vertices, a set of edges, a set of 2-simplices, etc., and then assigns to each edge its two vertices, assigns to each 2-simplex its three edges, etc.  The computer must only check that Formula \ref{def:set def of sset}(\ref{dia:semi identity}) holds.

\end{remark}

\begin{definition}\label{def:1-skeleton}

Let $X=\xymatrix@1{\cdots \arrrr{r}&X_2\arrr{r}&X_1\ar@<.5ex>[r]^{d_1}\ar@<-.5ex>[r]_{d_0}&X_0}$ be a semi-simplicial set.  The {\em 1-skeleton} of $X$ is the (directed) graph $$\xymatrix{X_1\ar@<.35ex>[r]^{d_1}\ar@<-.35ex>[r]_{d_0}&X_0,}\hspace{.2in}\tn{ or one could write } \hspace{.2in} X_1\To{d_1\cross d_0}X_0\cross X_0,$$ obtained by forgetting all higher-order ($n\geq 2$) data in $X$. 

\end{definition}

\begin{example}\label{ex:square}

A filled-in square, composed of two solid triangles glued along their diagonal, is drawn below.

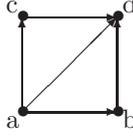
\begin{center}\begin{picture}(40,40)

\put(0,0){a}\put(0,44){c}\put(44,0){b}\put(44,44){d}
\put(4,4){$\bullet$}\put(4,40){$\bullet$}\put(40,4){$\bullet$}\put(40,40){$\bullet$}
\put(6,6){\vector(1,0){36}}\put(6,6){\vector(1,1){36}}\put(6,6){\vector(0,1){36}}
\put(42,6){\vector(0,1){36}}\put(6,42){\vector(1,0){36}}
\end{picture}\end{center}

Here we have $X_0=\{a,b,c,d\}, X_1=\{ab,ac,ad,bd,cd\}, X_2=\{abd,acd\}$, and $X_n=\emptyset$ for $n\geq 3$.  Each triangle has three edges, given by the various $d^2_i$'s: $d^2_0(abd)=bd, d^2_0(acd)=cd, d^2_1(abd)=d^2_1(acd)=ad, d^2_2(abd)=ab,$ and $d^2_2(acd)=ac$.  Similarly, each edge has two vertices: $d^1_1(ac)=d^1_1(ad)=d^1_1(ab)=a$, $d^1_0(ad)=d^1_0(bd)=d^1_0(cd)=d$, $d^1_0(ac)=d^1_1(cd)=c$, and $d^1_0(ab)=d^1_1(bd)=b.$ 

Knowing the above sets $X_i$ and face operator equations allows one to completely reconstruct the square.  They are precisely the data of the semi-simplicial set $X$.

The 1-skeleton of the above square is drawn the same way, but one imagines the triangles as empty in this case; it is just four vertices and five edges.  As a graph, it is written $\{ab,ac,ad,bd,cd\}\to\{a,b,c,d\}^2$, where the map assigns to each edge the pair consisting of its first and second vertex.  This is precisely what is achieved by $d^1_1$ and $d^1_0$ above.

\end{example}

\begin{definition}\label{def:cat of semi}

A {\em morphism of semi-simplicial sets} $f\taking X\to Y$ consists of maps $f_n\taking X_n\to Y_n$ for each $n\in\NN$ such that $f_n\circ d^{n+1}_i=d^{n+1}_i\circ f_{n+1}$ for all $0\leq i\leq n+1$.

With objects defined in Definition \ref{def:set def of sset} and morphisms defined above, and with the obvious composition law, the category of semi-simplicial sets has now been defined; we denote it $\mcS$.

\end{definition}

The set-theoretic definition of simplicial sets is much more involved than that of semi-simplicial sets, the difference being that $n$-simplices in a simplicial set can be ``degenerate," meaning they are really smaller simplices in disguise.  Instead of having to satisfy one identity (like the displayed formula (\ref{dia:semi identity})), a simplicial set must satisfy six.  For this reason, we do not include the set-theoretic definition of simplicial sets; it can be found in \cite{GJ}.  Luckily, the categorical definition of simplicial sets is no harder than that of semi-simplicial sets.   Compare Proposition \ref{prop:cat semi} and Definition \ref{def:sset}.

\subsection{Categorical definitions of graphs, hypergraphs, and simplicial sets}\label{subsec:cat defs}

In this section we give category-theoretic definitions of graphs, hypergraphs, semi-simplicial sets, and simplicial sets.  We find that there is a common category-theoretic framework which holds all three examples (and more, as we will see in the next section).  That is, we can consider each of them as a {\em category of functors}, from some ``indexing" category to the category of sets.  In the case of graphs the indexing category will have two objects, $V$ and $E_2$, and two maps $V\to E_2$; in the case of hypergraphs the indexing category will have (infinitely-many) objects $V,E_1,E_2,E_3,\ldots$ and $n$ maps $V\to E_n$ for each $n$; and in the case of semi-simplicial sets the indexing category will have objects $[0],[1],[2],\ldots$ and ${j+1}\choose{i+1}$ maps $[i]\to[j]$.

Recall that if $\mcC$ is a category then a functor $F\taking\mcC\to\Sets$ assigns to each object $C\in\Ob(\mcC)$ a set $F(C)$, and to each morphism $C\to D$ in $\mcC$ a function between sets $F(C)\to F(D)$.  If $G\taking\mcC\to\Sets$ is another functor, then a natural transformation $\alpha\taking F\to G$ is a coherent way to give a function $F(C)\to G(C)$ for each $C\in\Ob(\mcC)$.  See \cite{MacLane}.

The functors $\mcC\to\Sets$ are thus themselves the objects of a larger category, written $\Sets^\mcC$, where the morphisms in this category are the natural transformations between functors.  While this may be difficult to understand at first, one should remember that the purpose of $\mcC$ is simply to prescribe an arrangement of sets and functions between them.  The category-theoretic definition of graphs, given as Proposition \ref{prop:cat graph} below, may help orient the reader to this way of thinking.

\begin{proposition}[Categorical definition of graphs]\label{prop:cat graph}

Let $\bD_G$ denote the category with two objects $V$ and $E_2$, two morphisms $V\to E_2$, and no other non-identity morphisms.  Let $\bD_G\op$ denote the opposite category.  The category $\mcG$ of directed graphs is isomorphic to $\Sets^{\bD_G\op}$, the category of functors $\bD_G\op\to\Sets$.

\end{proposition}

\begin{proof}

Giving a functor $F\taking\bD_G\op\to\Sets$ consists of giving two sets, $F(V)$ and $F(E_2)$, and two functions $F(E_2)\tto F(V)$.  This is precisely the data of a graph $F_E\to F_V\cross F_V$ with vertex set $F_V:=F(V)$ and edge set $F_E:=F(E_2)$.  Giving a natural transformation $F\to F'$ consists of giving a function $F(V)\to F'(V)$ and a function $F(E_2)\to F'(E_2)$, such that the diagrams $$\xymatrix{F(E_2)\ar[r]\ar[d]&F(V)\cross F(V)\ar[d]\\F'(E_2)\ar[r]&F'(V)\cross F'(V)}$$ commute.  Such is precisely a morphism of graphs.

\end{proof}

One can give a similar category-theoretic definition of undirected graphs, using a slightly larger indexing category, but we will skip it.

\begin{proposition}[Categorical definition of hypergraphs]\label{prop:cat hyper}

Let $\bD_H$ denote the category with objects denoted $V, E_1,E_2,E_3,\ldots$ and with distinct morphisms $v_i^n\taking V\to E_n$ for each $1\leq i\leq n$, and no other non-identity morphisms.  Let $\bD_H\op$ denote its opposite.   The category $\mcH$ of directed hypergraphs is isomorphic to $\Sets^{\bD_H\op}$, the category of functors $\bD\op_H\to\Sets$.  

\end{proposition}

\begin{proof}[Sketch of proof]

Let $F\taking\bD_H\op\to\Sets$ denote an object in $\Sets^{\bD_H\op}$.  As a directed hypergraph, its vertex set is $F_V:=F(V)$, its hyperedge set is $F_E:=\coprod_{n\in\NN}F(E_n)$, and its constituent function is given by the product of the $n$ maps $F(v^n_i)\taking F(E_n)\to F(V)$.  One checks that a natural transformation $F\to F'$ is the same data as the respective morphism of hypergraphs, as in Definition \ref{def:cat of hyper}.

\end{proof}

As promised, we also give a category-theoretic definition of semi-simplicial sets.  For each $n\in\NN$, let $[n]$ denote the ordered set $\{0,1,\ldots,n\}$; note that each $[n]$ is finite and non-empty.  A function $f\taking[m]\to[n]$ is called {\em non-decreasing} if $j\geq i$ implies $f(j)\geq f(i)$, for all $0\leq i,j\leq m$.  The function $f$ is called {\em strictly increasing} if $j>i$ implies $f(j)>f(i)$.  The category whose objects are the ordered sets $[n]$ for $n\in\NN$ and whose morphisms are non-decreasing functions (respectively strictly-increasing functions) is denoted $\bD_\geq$ (respectively $\bD_>$).  

\begin{proposition}[Categorical definition of semi-simplicial sets]\label{prop:cat semi}

Let $\bD_>$ denote the category whose objects are the finite non-empty ordered sets $\{0,1,\ldots,n\}, n\geq 0$, and whose morphisms are strictly increasing functions.  Let $\bD_>\op$ denote the opposite category.  The category $\mcS$ of semi-simplicial sets is isomorphic to $\Sets^{\bD_>\op}$, the category of functors $\bD_>\op\to\Sets$.

\end{proposition}

\begin{remark}

Given a semi-simplicial set $X\taking\bD_>\op\to\Sets$, we write $X_n\in\Sets$ to denote the set $X(\{0,1,\ldots,n\})$; its elements are the {\em $n$-simplices of $X$}.  There are $n+1$ possibilities for strictly increasing maps $\delta^i_n\taking \{0,1,\ldots,n-1\}\to\{0,1,\ldots,n\}$, given by skipping exactly one of the numbers $0\leq i\leq n$.  Since $X$ is a contravariant functor, each of these morphisms is sent to a map $X(\delta_n^i)\taking X_n\to X_{n-1}$, which is the {\em $i$th face operator} and denoted $d^n_i$.  These face operators satisfy the properties given in the set-theoretic Definition \ref{def:set def of sset}.

\end{remark}

\begin{remark}\label{rem:drawbacks of complexes}

Strictly speaking, Atkins theory of Q-analysis is often presented in terms of simplicial complexes rather than simplicial sets or semi-simplicial sets (\cite{Atk}).  A simplicial complex is a semi-simplicial set in which no two simplices have the same (ordered) set of vertices.  So, for example, the graph with two nodes and two edges depicted $$
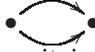
\xymatrix@1{\bullet\ar@/^.7pc/[r]\ar@/_.7pc/[r]&\bullet}$$ is {\em not} a simplicial complex.  However, it is a (semi-) simplicial set.  Taking the position that our models should generalize graphs, we must use simplicial sets rather than simplicial complexes.  

In applications, one might want to consider two research papers with the same set of authors as two different simplices linking the same vertices.  For this, one must use (semi-) simplicial sets, not simplicial complexes.

\end{remark}

\begin{definition}[Categorical definition of simplicial sets]\label{def:sset}

Let $\bD_\geq$ denote the category whose objects are the finite non-empty ordered sets $\{0,1,\ldots,n\}, n\geq 0$, and whose morphisms are non-decreasing functions.  Let $\bD_\geq\op$ denote the opposite category.  The {\em category of simplicial sets}, denoted $\sSets$, is the category of functors $\Sets^{\bD_\geq\op}$, whose objects are functors $\bD_\geq\op\to\Sets$ and whose morphisms are natural transformations.

\end{definition}

\comment{

\begin{remark}\label{rem:all you need to know}

This is not a technical paper designed to prove difficult results, but instead an attempt to convince the reader that simplicial sets are better models for certain networks.  Hence, we do not require that the reader understand the rigorous underpinnings of simplicial sets.  If convinced of their usefulness, the reader can easily learn about them on his or her own time.  See \cite{Fri} for a nice exposition.  We hope the reader simply understands the following.  

A simplicial set is a geometric object with corners, edges, triangular faces, tetrahedra, etc., which are called simplices.  An $n$-simplex has dimension $n$, so that a $0$-simplex is a vertex, a $1$-simplex is an edge, a $2$-simplex is a triangle, and an $n$-simplex is an $n$-dimensional triangle.  Each $n$-simplex has lots of sub-simplices, as a triangle has edges and vertices.  All simplices have ordered vertices, and the vertices of a sub-simplex are in the same order as in the simplex.

If simplicial sets are geometric objects comprised of these triangles, then we must be able to fit together these triangles in some ways.  To do so, simply line up some vertices of one simplex with some vertices of another; doing so will match up a face of the first simplex with a face of the other.  Then attach these two simplices along that common face.  In Example \ref{ex:square} we showed a square; it can be obtained by an attaching process as we have described here, by gluing the triangle $abd$ to the triangle $acd$ along the common edge $ad$.

\end{remark}

}

\begin{remark}\label{rem:difference semi vs. simp}

We promised to compare simplicial sets and semi-simplicial sets.  These notions have had a long history, and it took many years to sort out terminology (see \cite{Jam}).  Both are useful, although simplicial sets have become dominant in algebraic topology. 

To compare semi-simplicial sets and simplicial sets, one should begin by comparing their ``indexing categories," $\bD_>$ and $\bD_\geq$, the difference being the existence of non-injective morphisms in $\bD_\geq$.  Geometrically, this allows one to consider simplices (triangles) in a simplicial set that are {\em degenerate}.  A degenerate $n$-simplex is an $n$-dimensional triangle that has been flattened onto one of its $(n-k)$-dimensional faces.  The use of these degenerate simplices has two advantages from a topological perspective: it is easier to build most shapes using simplicial sets (e.g. the 2-sphere), and products of simplicial sets are better behaved under geometric realization (see \cite{Hov}).  See also Remark \ref{rem:power of ssets} for a diverse list of mathematical contexts in which simplicial sets naturally arise, such as algebra and $\infty$-category theory.

\end{remark}

\subsection{Comparing hypergraphs and simplicial sets}\label{subsec:comparisons}

We have found that graphs, hypergraphs, and (semi-) simplicial sets are all examples of the same phenomenon: each is a category of functors from some indexing category to $\Sets$, meaning each is simply a prescribed arrangement of sets.  This discovery leads to an increased ability to change between one model and another; one needs only find a relationship between the categories that index them.  Let us make this more precise.

\begin{lemma}\label{lemma:indexed adjunction}

Let $\mcC$ and $\mcD$ be categories.  Given a functor $A\taking\mcC\to\mcD$, one obtains a functor $R_A\taking\Sets^\mcD\to\Sets^\mcC$ given by $R_A(G)=G\circ A$, for any $G\in\Sets^D$.  There is also a functor $L_A\taking\Sets^\mcC\to\Sets^\mcD$, and a natural bijection $$\Hom_{\Sets^\mcC}(F,R_A(G))\To{\iso}\Hom_{\Sets^\mcD}(L_A(F),G)$$ for any $F\in\Sets^\mcC, G\in\Sets^\mcD$.

\end{lemma}

\begin{proof}

See \cite[Thm I.5.2]{MM}.

\end{proof} 

To bring this down to earth, recall the categories $\Delta_G$ and $\Delta_H$, which respectively index graphs and hypergraphs (see Propositions \ref{prop:cat graph} and \ref{prop:cat hyper}).  There is an obvious functor $\Delta_G\to\Delta_H$ sending $V\mapsto V$ and $E_2\mapsto E_2$; this functor is just the inclusion of $\Delta_G$ as a subcategory of $\Delta_H$.  It also induces a functor on opposite categories which we denote $i\taking\Delta_G\op\to\Delta_H\op$.  Thus we may apply Lemma \ref{lemma:indexed adjunction} to obtain functors $$\Adjoint{L_i}{\mcG}{\mcH}{R_i}.$$  The functor $R_i$ takes a hypergraph and forgets everything but the vertices and the 2-edges, to yield a graph.  The functor $L_i$ takes a graph and considers it as a (2-uniform) hypergraph.  Note that $R_i\circ L_i\taking\mcG\to\mcG$ takes any graph to itself, but the functor $L_i\circ R_i\taking\mcH\to\mcH$ destroys information.

We move on to compare hypergraphs and simplicial sets in the same way, by supplying a functor between indexing categories $\bD_H\to\bD_\geq$ as follows.  On objects, send $V\mapsto [0]$, and for each $n\geq 1$, send $E_n\mapsto [n-1]$.  On morphisms, send $v^n_i\taking V\to E_n$ to the morphism $(0\mapsto i-1)\taking [0]\to[n-1]$, for each $1\leq i\leq n$.  Let $A\taking\bD_H\op\to\bD_\geq\op$ denote its opposite.  By Lemma \ref{lemma:indexed adjunction}, we obtain functors $$\Adjoint{L_A}{\mcH}{\sSets}{R_A}.$$  The functor $R_A$ takes a simplicial set and forgets all relationships between $n$-simplices and $m$-simplices, except when $m=0$.  The functor $L_A$ takes a hypergraph $H$ and sends it to a simplicial set whose $0$-simplices are the vertices of $H$ and whose $n$-simplices are the $(n+1)$-edges of $H$ for $n\geq 1$.  Note that the composition $R_A\circ L_A\taking\mcH\to\mcH$ is not quite the identity: a hypergraph is replaced by the one in which every vertex is considered a 1-edge (and everything else stays the same).  Some authors (\cite{Weis}) define hypergraphs as not having 1-edges to begin with; for them the composite $R_A\circ L_A$ would be the identity.  What these authors call hypergraphs, we call ``basic higher graphs" -- see Definition \ref{def:cat of higher graphs}.

Similarly, there are functors between simplicial sets and semi-simplicial sets, between graphs and (semi-) simplicial sets, etc.  Any functor between indexing categories gives rise to functors between the categories they index.  This is the basic idea behind the next subsection, where we define what it means to be a category of higher graphs.

\subsection{Categories of higher graphs}\label{subsec:cats of higher graphs}

We saw above that category theory provides a framework in which graphs, hypergraphs, simplicial sets, and semi-simplicial sets can all be defined.  One naturally asks: what other categories might be useful for modeling networks?

In order to answer this, we must decide on what we want a theory of higher-dimensional graphs to have.  We clearly want a set of vertices and sets of (higher) edges, such that each higher edge has an underlying set of vertices.  It seems that having a set of 1-edges is not essential, so we do not require it.  

\begin{definition}\label{def:indexing cats}

Let $\bD_B$ denote the category with an object $V$ and objects $\{E_n|n\geq 2\}$, and with morphisms $v^n_i\taking V\to E_n$ for each $1\leq i\leq n$.  We call $\bD_B$ the {\em basic indexing category for higher graphs} (or just {\em basic indexing category} for short). 

An {\em indexing category for higher graphs} (or {\em indexing category} for short) is a pair $(\mcI,\bD_B\To{f}\mcI)$, where $\mcI$ is a category and $f\taking\bD_B\to\mcI$ is a functor.  We may denote $(\mcI,\bD_B\To{f}\mcI)$ simply by $\mcI$, if $f$ is clear from context.

We say that an indexing category $(\mcI,\bD_B\To{f}\mcI)$ is {\em faithful} if $f$ is a faithful functor.

\end{definition}

\begin{remark}\label{rem:three models}

An indexing category for higher graphs is just a category which has an ``interpretation" of vertices and $n$-edges for each $n\geq 2$.  It is faithful if the $n$ vertices of a generic $n$-edge are distinct.  Thus each of $\bD_H, \bD_>,$ and $\bD_\geq$ is a faithful indexing category for higher graphs in an obvious way.  In fact, even $\bD_G$ is an indexing category for higher graphs (in many different ways), because there are infinitely many functors $\bD_B\to\bD_G$.  However, none of these are faithful, so the indexing category for graphs cannot be considered as a faithful model for higher graphs.

The above paragraph focuses on indexing categories for higher graphs.  After making the following definition (\ref{def:cat of higher graphs}), we will be able to say that hypergraphs, semi-simplicial sets, and simplicial sets are all faithful categories of higher graphs.

\end{remark}

\begin{definition}\label{def:cat of higher graphs}

The category of functors $\bD_B\op\to\Sets$, denoted $\mcB=\Sets^{\bD_B\op}$, is called the {\em basic category of higher graphs}.  The objects of $\mcB$ are called {\em basic higher graphs} (or just {\em basic graphs} for short).

A {\em category of higher graphs} (or just {\em category of graphs} for short) is a category of the form $\Sets^{\mcI\op}$, where $(\mcI,\mcB\To{f}\mcI)$ is an indexing category for higher graphs.  We say that it is {\em faithful} if $f$ is.  We call $\Sets^{\mcI\op}$ the {\em category of $\mcI$-indexed graphs}.

Given an indexing category $(\mcI,\mcB\To{f}\mcI)$ there is always a pair of functors $$\Adjoint{L_f}{\mcB}{\Sets^{\mcI\op}}{R_f},$$ in accord with Lemma \ref{lemma:indexed adjunction}, comparing the category of $\mcI$-indexed graphs and the category of basic graphs.  For any $\mcI$-indexed graph $G\taking\mcI\op\to\Sets$, we call $R_f(G)$ its {\em underlying basic graph}.

\end{definition}

\begin{example}[Undirected higher graphs]\label{ex:symmetric}

Let $\Sigma$ denote the category whose objects are the sets $\underline{n}:=\{0,1,\ldots,n\}$, one for each $n\in\NN$, and whose morphisms are simply functions.  Thus there are $(n+1)^{(m+1)}$ morphisms $\ul{m}\to\ul{n}$.  Note that $\Sigma$ is a faithful indexing category for higher graphs via the obvious functor $\mcB\to\Sigma$ sending $V\mapsto \ul{0}$ and $E_n\mapsto\ul{n-1}$ for $n\geq 2$.  The category $\Sets^{\Sigma\op}$ of $\Sigma$-indexed graphs may also be called the category of {\em symmetric simplicial sets}.  See \cite{Gra}.  

One can visualize a symmetric simplicial set as a simplicial set in which the vertices of a given simplex are unordered.  It is the union of vertices, edges, triangles, etc.  Symmetric simplicial sets are the higher-dimensional analogue of undirected graphs, so we sometimes call them ``undirected higher graphs."

\end{example}

\begin{example}[Broadcasting graphs]\label{ex:broadcasting}

For any $n\in\NN$, let $(n)$ denote the category with an initial object $R_n$ and $n$ other objects, and that has no non-identity morphisms except for the unique morphism from $R_n$ to each other object.  One can count that there are $(n+1)^m+n$ functors $(m)\to (n)$.  Let $T$ denote the category whose objects are the categories $(n)$ and whose morphisms are the functors between them.  We call $T$ the {\em indexing category for broadcasts}.

There is an obvious functor $\mcB\to T$ given by $V\mapsto (0)$ and $E_n\mapsto (n-1)$ for each $n\geq 2$.  In fact this functor is faithful, so $T$ is a faithful indexing category for higher graphs.  

One may think of a $T$-indexed graph $F\taking T\op\to\Sets$ as a model for broadcasting and rebroadcasting.  The set $T(n)$ is the set of broadcasts that reach $n$ other entities.  Morphisms $(m)\to(n)$ allow one to connect various broadcasts together to form a network.  We call $T$-indexed graphs {\em broadcasting graphs}.

\end{example}

\subsection{Changing between models}

One goal of this paper is to allow researchers to change between the various models of higher graphs.  In this section we explain how this can be done.  The only necessary ingredient is a morphism of indexing categories.  

\begin{definition}

Let $(\mcI,\mcB\To{f}\mcI)$ and $(\mcJ,\mcB\To{g}\mcJ)$ denote indexing categories for higher graphs.  A {\em morphism of indexing categories} is functor $h\taking\mcI\to\mcJ$ such that $h\circ f=g$.  

\end{definition}

For the following Lemma, we assume the reader is familiar with adjoint functors.  We have implicitly given the definition already in Lemma \ref{lemma:indexed adjunction}, but one may also see \cite{MacLane}.

\begin{proposition}\label{prop:two adjunctions}

Let $\mcI=(\mcI,\mcB\To{f}\mcI)$ and $\mcJ=(\mcJ,\mcB\To{g}\mcJ)$ denote indexing categories.  Given a morphism of indexing categories $h'\taking\mcI\to\mcJ$, there exists two adjunctions \begin{align}\Adjoint{L_h}{\Sets^{\mcI\op}}{\Sets^{\mcJ\op}}{R_h}\\\Adjoint{R_h}{\Sets^{\mcJ\op}}{\Sets^{\mcI\op}}{\forall_h}\end{align} where $h=(h')\op$.

\end{proposition}

\begin{proof}

See \cite[7.2.2]{MM}.

\end{proof}

\begin{example}

Let $\mcI=\Delta_G$ and $\mcJ=\Delta_>$ denote the indexing categories for graphs and semi-simplicial sets, respectively, and let $h\taking\mcI\op\to\mcJ\op$ denote the functor $V\mapsto [0], E_2\mapsto [1]$.  We will explain the functors $R_h,L_h$, and $\forall_h$ from Proposition \ref{prop:two adjunctions} in this case.

The functor $R_h$ takes a semi-simplicial set $X$ and returns the graph with vertices $X_0$ and edges $X_1$.  The functor $L_h$ takes a graph $Y$ and returns the semi-simplicial set with $0$-simplces $Y_V$, 1-simplices $Y_E$, and no higher simplices.  

The most interesting of the three is $\forall_h\taking\mcG\to\mcS$.  It takes a graph $Y$ and returns the semi-simplicial set with an $n$-simplex connecting every set of $n+1$ vertices which forms a complete graph (or {\em clique}) in $Y$.  
  
\end{example}

\section{Choosing the appropriate model}\label{sec:choosing}

As discussed in the introduction, graphs are not suitable models for many kinds of networks that exist in the real world.  One way of stating the inadequacy is that graphs do not contain enough data to discern multi-person conversations from private ones.  They do not leave room for contexts shared by more than two people.  

For example, 3-way calling or conference calls cannot be adequately modeled by graphs.  Information is transfered differently in a 3-way call between $A, B,$ and $C$ than it is if the three people can only talk in two way conversations.  The world of conference calls can be modeled by hypergraphs, with an edge connecting any set of vertices that are engaged in a given call.  It would not be wrong to model this with situation with simplicial sets, but doing so will not really contribute anything of value.  However, one may imagine adding a feature to conference calls in which any participant can, if he or she chooses, speak more privately to any subset of the other participants.  This second scenario is nicely modeled by simplicial sets (in the case of our 3-way call, it is modeled by a 2-simplex with vertices $A,B,C$).

Imagine that an intelligence agency (like the CIA) is trying to listen to the above three-way conversation, but for some reason they can only listen to the communication between $B$ and $C$.  This scenario is modeled by a simplicial set with two 2-simplices, $ABC$ and $BCI$, where $I$ represents the intelligence agency.  These two 2-simplices are glued together along the common edge $BC$.  Such a scenario cannot be adequately modeled by hypergraphs, because of the ``sub-edge problem" detailed below in Remark \ref{rem:sub-edge problem}. 

Recall that one goal of this paper is to offer many different models available representing higher graphs.  In this section, we try to make clear how one should choose whether to use graphs, hypergraphs, simplicial sets, symmetric simplicial sets, etc. to model the scenario at hand.  Clearly, one wants to choose the model whose indexing category is as simple as possible, under the condition that it accurately reflects all of the features that one wishes to study.

\subsection{Networks modeled by graphs}

Many networks can be modeled by (1-dimensional) graphs.  When a person sends email, it always passes privately to the recipient.  Even if the recipient list of an email includes multiple addresses, none of the recipients has any way to truly know that the same message was sent to each recipient.  The message does not exist in a shared space.

Graphs are good models for any type of network in which all correspondences are private, or ``one-on-one," because every incidence between nodes involves only two nodes.  The graph model breaks down as soon as one wants to consider threesome, foursomes or $n$-somes interacting all at once.  In an online social network (e.g. Facebook), some communications are posted to a public forum which can be accessed by any member of the group.  Assuming the internal integrity of the system (i.e. assuming Facebook displays the same messages to each member of the group), the message is {\em shared} by the group.  These situations are not well-modeled by graphs.  

The following are some networks that are best modeled by graphs (\cite{BG}).

\begin{enumerate}

\item transportation networks (locations, roads);
\item tournaments (players, one-on-one matches);
\item electrical flow (terminals, wires);
\item private gossip (people, one-on-one conversations).

\end{enumerate}

What these networks have in common is that every interaction is between exactly two nodes.  For this reason, the indexing category $\Delta_G$, itself having only two non-identity arrows $d_0,d_1$, is appropriate.

\subsection{Networks modeled by hypergraphs}

In many applications, there may be groups that interact in a way that subgroups do not; hypergraphs are often appropriate models for these situations.  As we mentioned above, in a conference call any person can speak to the whole group but cannot speak to a given subgroup.  In Facebook, one may post a message to an entire group of which he or she is a part, but not to an arbitrary subgroup of it.  

As another example, one may model the academic research network by a hypergraph in which every researcher is a vertex, and every coauthorship is a hyper-edge.  If three researchers are the coauthors on a paper, that does not imply that some subset of them is also the set of coauthors of a paper.  Coauthorship is a relation that is not ``closed under taking subsets" (see Definition \ref{def:closed under subsets}).  For this reason hypergraphs are the appropriate model of coauthorship.  Similarly, the network consisting of sets of people who have eaten lunch as a group should be modeled by hypergraphs, because this relation is not closed under taking subsets.

\begin{definition}\label{def:closed under subsets}

A {\em symmetric multi-relation} on a set $S$ consists of a subset $R\subseteq\PP(S)$, i.e. a set of subsets of $S$.  We that $R$ is {\em closed under taking subsets} if, whenever $x\subseteq S$ is an element of $R$ and $x'\subseteq x$ is a subset of it, then $x'\in R$ as well.

\end{definition}  

\begin{remark} 

There is also an analogous notion of (non-symmetric) multi-relation $R\subset\amalg_{n\in\NN}S^n$, and of such a multi-relation being closed under taking subsets, but we will not define the latter here.

\end{remark}

We can see that multi-relations can always be modeled by hypergraphs.  If the multi-relation is closed under taking subsets, then semi-simplicial sets (or symmetric simplicial sets, etc.) can also be used.  It is often useful to change ones perspective on a relationship so that it becomes closed under taking subsets.  Let us show what we mean by taking another look at the examples above.

Instead of modeling a ``coauthorship" network, change perspectives and model a ``collaboration" network: if $n$ people have all worked together on a project then so has every subgroup.  Or, instead of looking at the network of sets of people who have had lunch {\em as a group}, look at the network of people who have had lunch {\em together}.  The same goes for conference calls and Facebook groups: if a group of people are involved together, then so is any subgroup.  

In all of these cases, one can change perspectives on what the network is, and then use semi-simplicial sets as a model instead of hypergraphs.  To do so, one would simply apply the functor $L_f$ from Proposition \ref{prop:two adjunctions} (for the unique morphism of indexing categories $f\taking\Delta_H\to\Delta_>$ that sends $E_n\mapsto [n-1]$).

If the network exists by way of technology, the simplicial model may suggest an additional feature for that technology.  For example, as mentioned above, one might want to add a feature to conference calling whereby a participant can speak privately to a subset of the other participants (and similarly for Facebook groups).

Connecting groups together is where (semi-) simplicial sets have a big advantage over hypergraphs.  If two members of a collaborative effort take their knowledge and apply it to another collaborative effort, we may want to imagine this as two simplices joined together along a common edge in the collaboration network.  We show below that hypergraphs have a serious drawback in this context, which we have not yet made explicit.  We call it the sub-edge problem.  

\begin{remark}[Sub-edge problem]\label{rem:sub-edge problem}

Suppose that $X$ is a hypergraph with vertices $\{a,b,c,d\}$ and three edges $e_1=(a,b),e_2=(a,b,c),e_3=(a,b,c)$, so that two edges have the same vertices.  If one were to ask whether $e_1$ is a subedge of $e_2$ or $e_3$, either we would decline to answer (in Japanese we would say {\em mu}: ``unask the question"), or we would be required to say that $e_1$ is a subedge of both $e_2$ and $e_3$.  Using hypergraphs, there is no good notion of sub-edges.  

This issue is clear when looking at the underlying indexing category for hypergraphs, $\bD_H$.  It has no morphisms $E_n\to E_m$ for any $m\neq n$.  Thus a hypergraph $X\taking\bD_H\op\to\Sets$ has no capability for comparing the sets $X(E_m)$ of $m$-edges and $X(E_n)$ of $n$-edges.  

We explain the problem from one more angle.  Suppose that $X'$ is the hypergraph with the same vertices as $X$ but with a slight change in the edge sets: $e'_1=(a,b),e'_2=(b,c),e'_3=(a,b,c)$.  There is {\em no} morphism $X'\to X$ respecting the vertices, even though $e'_2$ looks like a subedge of $e_2$.   

As an application, suppose we are modeling groups in Facebook, for example, and user $a$ joins group $e_2'$.  This event cannot be modeled by a morphism of hypergraphs. 

Hypergraphs have no good notion of sub-edges, but semi-simplicial sets do (have a good notion of sub-simplices).  Thus, the issues in this remark go away when dealing with semi-simplicial sets, or any model of higher graphs whose indexing category has an adequate set of morphisms.

\end{remark}

\subsection{Network modeled by semi-simplicial sets}\label{subsec:ssets}

If we look back at the definition of semi-simplicial sets (Definition \ref{def:set def of sset}), we notice that each simplex has ordered faces and hence ordered vertices.  In many of our discussions above, we have been suppressing this fact.  Semi-simplicial sets may ``look like" the polygonal shapes we have described, but they also have more structure.    This is akin to the fact that graphs can be directed; they may look like 1-dimensional spaces, but each edge has ordered vertices, giving it a directionality.  

To represent a given real-world situation, this ordered structure may be useful or it may not.  If it is not, we could either continue to use semi-simplicial sets as our model but forget the ordered structure, or we could use symmetric (semi-) simplicial sets (see Example \ref{ex:symmetric}).  In this subsection we discuss situations in which semi-simplicial sets, in all their ordered glory, are more appropriate as models than their symmetric analogues are.  In Subsection \ref{subsec:other networks} we explain situations in which other models are more appropriate.

Just like for directed graphs (versus undirected graphs) the ordering on simplices is useful for modeling a network if, within any given interaction of that network, the nodes naturally have an order!  For example, imagine a multi-cast network in which one node sends a message to a list of other nodes, with the provision that every node in the list is also responsible for the successful delivery of that message to all subsequent nodes.  This scenario is best modeled by semi-simplicial sets because the nodes in each simplex are ordered.

As another example, suppose that for every military mission, the soldiers involved are strictly ranked, so that each soldier either outranks or is outranked by each other soldier.  A scenario in which three soldiers $A,B,C$ are simultaneously involved in two different missions, one consisting of five soldiers and the other consisting of six soldiers, would be modeled by the union of a 4-simplex and a 5-simplex, glued together along a common 2-simplex ($ABC$).  The only condition is that the rank of soldiers $A,B,$ and $C$ is the same in both missions.

As a final example, consider a set of interacting processes, say in a computer program.  Suppose that each process has a finite number $n$ of parts which should occur in a fixed order, but that some process $P$ may be a subprocess of two distinct processes $Q,R$.  We would model such a scenario as a semi-simplicial set (in this case with $Q$ and $R$ glued together along $P$).

Just like directed graphs can be used to model undirected situations (one includes edges in both directions for every pair of connected nodes), semi-simplicial sets can be used to model situations that are not inherently ordered.  Again, one simply includes all $(n+1)!$ permutations of any given $n$-simplex. 

Semi-simplicial sets are expressive and yet combinatorial (i.e. they fall under the banner of ``discrete mathematics").  Even more so are their cousins, simplicial sets; we record their wide-ranging mathematical application in the following remark.

\begin{remark}\label{rem:power of ssets}

Simplicial sets are widely used in category theory, algebra, and algebraic topology.  In category theory, any monad on a category gives rise to a natural simplicial resolution of objects in that category.  Thus simplicial sets come up as resolutions of objects in an algebraic category such as groups or algebras, and hence are indispensable in homology theory \cite{Wei}.  They also are sufficient as models for topological spaces up to homotopy \cite{Hov}.  Finally, new work by Joyal \cite{Joy} and by Lurie \cite{Lur} shows that simplicial sets also provide a good model for the theory of $\infty$-categories.

These applications and others may also prove useful or interesting to a computer science researcher.

\end{remark}
 
\subsection{Other networks}\label{subsec:other networks}

Many real-world networks involve interactions between more than two people, as well as a certain kind of interaction between interactions.  Each of these networks should be modeled by a higher graph, but the choice of indexing category is a bit of an art.  

We mentioned networks involving hierarchies in Subsection \ref{subsec:ssets}, using military operations as one example.  However, because the objects of $\bD_\geq$ are strictly ordered sets, we had to assume that the participants in any operation are strictly ranked in order to use simplicial sets as a model.  If instead some participants are ranked equally, then we may choose a different model, i.e. a category of higher graphs with an indexing category whose objects are not so strictly ranked.

The issue is really the anti-symmetric law: $x\leq y$ and $y\leq x$ implies $x=y$.  This holds for any two elements $i,j$ of $[n]\in\bD_\geq$.  Let us relax that restriction.

The category of non-empty finite linear preorders, denoted $\bD_{PrO}$ has sets $X$ as objects, where $X$ has elements $x_1,\ldots x_n$ (for some $n\geq 1$), and an ordering $\succeq$ such that $i\geq j$ implies $x_i\succeq x_j$.  This ordering is not required to be anti-symmetric, but note that as we have defined it, $\succeq$ is automatically reflexive and transitive.  A morphism of these linear pre-orders is just a non-decreasing function.  Note that $\bD_\geq$ is a subcategory of $\bD_{PrO}$; hence $\bD_{PrO}$ can be given the structure of an indexing category for higher graphs.  

The category of $\bD_{PrO}$-indexed graphs is the richest category of higher graphs we have discussed so far.  Its objects look like simplicial sets, except that vertices in any simplex are not strictly ordered.  If, for example, three soldiers $A,B,C$ are ranked equally when working in one operation, but unequally in another, the simplices representing the two operations may be glued together along the ordered 2-simplex.  However, if the three soldiers are ranked $A\succ B\succ C$ in one operation and $A\succ C\succ B$ in another, then the two operations cannot be glued.

There are infinitely many categories that can serve as indexing categories for higher graphs (see Definition \ref{def:indexing cats}).  One chooses the appropriate indexing category by determining the key features of how interactions fit together.

\section{Applications and further research}\label{sec:apps}

In this paper we have given a wide range of mathematical objects which can be used to model higher graphs and networks.  We have shown that one can transform one model into another in a formal way, and tried to make clear how to choose the model which best suits the scenario at hand.  For example, we have discussed conference calls, collaboration networks in academic research, military missions, interlocking computer processes, and more.

To apply these ideas, one might examine real-world scenarios and their existing models, and ask whether another model would be more appropriate.  For example, one might realize that the broadcasting model (Example \ref{ex:broadcasting}) better captures the key features of a certain situation which is currently being modeled with hypergraphs.  If that researcher still has access to the original data, it should be possible to simply create the broadcasting graph; if not, he or she can apply one of the two functors given in Proposition \ref{prop:two adjunctions} to create a ``best guess" for the broadcasting graph associated to the given hypergraph.

Let us give just one more application, before turning to further research.  Consider the network whose nodes are computer programs, and where we consider programs to be interacting if they are outputting data with respect to the same protocols.  In other words, we have an $n$-simplex connecting any set of $n+1$ programs that have agreed to output data under a certain set of protocols.  Simplices overlap when programs agree to more than one set of protocols.  This situation is {\em not} well-modeled by either a simplicial complex (because the same set of programs may be interacting in two different ways), nor by hypergraphs (see the sub-edge problem, Remark \ref{rem:sub-edge problem}).  Instead, the above simplicial set should be examined in various cases to determine how the topology, such as homotopy groups, of a network may provide useful information.

Further research on the subject could go in many possible directions.  For example, one could attempt to provide a systematic way to classify networks that emerge in, say sociology or economics, and determine which category of higher graphs is best suited for each one.  

Another avenue of research is to use fuzzy simplicial sets, rather than just simplicial sets, to model networks.  A fuzzy set is a set in which each element exists only to a certain degree, a number between 0 and 1.  Fuzzy sets form a {\em topos} (\cite{Bar}), which roughly means that, as a category, they are as well-behaved as sets are.  Fuzzy simplicial complexes have been studied before (\cite{EER}), but it does not seem that they have been used in the study of networks.  One should consider a fuzzy simplicial set as a kind of weighted graph.  The idea would be that every vertex, edge, 2-simplex, etc. would have a certain strength, measuring the connectedness of the group of participants.  Making all this precise is not hard, and applications should be numerous \cite{Spi2}.

Another question to be studied is ``how does a network take in data and turn it into information?"  Humans communicate with one-another, either privately or in groups, and pass along ideas.  In doing so, we create our society from within.  If the society is considered a network modeled by a simplicial set, where every conversation is a simplex, then we are interested in how these conversations fuse together to create the zeitgeist.  To apply mathematics to this question, one would need to find a way to pose its conceptual core in mathematical language.  Once done, the question becomes a type of local-global calculation, which is one of the main topics of study in algebraic topology.  Future research should include the use of topological techniques to examine societal questions.

One could also use category-theoretic techniques to study these same societal questions.  Consider a party, modeled as a simplicial set in which the vertices represent people, and in which $n$-simplices represent sets of $n+1$ people who can hear each other.  This simplicial set $P$ is well-suited to understand how messages are transfered at this party, but not how those messages are interpreted.  But suppose we also had access to something like a ``category of vocabularies," where a morphism from one object to another represents how phrases in one person's vocabulary are interpreted in another's.  In this case, there should be a functor from the ``categorification" of $P$ to the category of vocabularies.  If we can make all this precise, we will have a nice model for how a network processes information.

\bibliographystyle{amsalpha}
\bibliography{networks-biblio}

\end{document}